\documentclass{article}

\usepackage[auth-sc]{authblk}

\usepackage{graphicx}
\usepackage{amssymb,amsmath,latexsym,color}
\usepackage[mathscr]{eucal}

\newtheorem{theorem}{Theorem}[section]
\newtheorem{lemma}[theorem]{Lemma}

\newtheorem{remark}[theorem]{Remark}
\newenvironment{proof}{\noindent
  \textbf{Proof.}}{\hfill$\Box$\\}

\newcommand{\instr}[5]{\ensuremath{\hbox to 60 pt
    {${#1}$\hfil${#2}$\hfil$ \rightarrow
      $\hfil${#3}$\hfil${#4}$\hfil${#5}$}}}

\newcommand{\vp}{\ensuremath{\varphi}}

\newcommand{\union}{\, \cup \,}

\newcommand{\CTL}{{\bf CTL}}

\newcommand{\CTLstar}{\ensuremath{{\bf CTL}^{\bf *}}}

\newcommand{\con}{\wedge}

\newcommand{\dis}{\vee}

\newcommand{\imp}{\rightarrow}

\newcommand{\equivalence}{\leftrightarrow}
\newcommand{\bottom}{\perp}

\newcommand{\Var}{\textit{\texttt{Var}}}

\newcommand{\truth}{\ensuremath{\top}}
\newcommand{\falsehood}{\ensuremath{\bot}}

\newcommand{\cm}[1]{\ensuremath{\sf{#1}}}

\newcommand{\next}{\!\raisebox{-.2ex}{ 
            \mbox{\unitlength=0.9ex
            \begin{picture}(2,2)
            \linethickness{0.06ex}
            \put(1,1){\circle{2}} 
            \end{picture}}}       
            \,}
          \newcommand{\until}{\ensuremath{\hspace{2pt}\mathcal{U}}}

\newcommand{\mmodel}[1]{\ensuremath{\frak{#1}}}
\newcommand{\states}[1]{\ensuremath{\mathcal{#1}}}

\newcommand{\sat}[3]{\ensuremath{\frak{#1}, #2 \models #3}}

\newcommand{\notsat}[3]{\ensuremath{\mmodel{#1}, #2 \not\models #3}}
\newcommand{\ATL}{{\bf ATL}}
\newcommand{\ATLstar}{\ensuremath{{\bf ATL}^{\bf *}}}

\newcommand{\coal}[1]{\ensuremath{\langle \hspace{-2.2pt} \langle #1
    \rangle \hspace{-2.5pt} \rangle}}

\newcommand{\psmodel}[1]{\ensuremath{\mathcal{M}^*}}
\newcommand{\pseudomodel}[1]{\ensuremath{\mathcal{M}^{**}}}

\newcommand{\AX}{\ensuremath{\forall \next}}
\newcommand{\EX}{\ensuremath{\exists \next}}
\newcommand{\AG}{\ensuremath{\forall \Box}}

\newcommand{\ar}{\ensuremath{\longmapsto}}
\newcommand{\sameas}{\ensuremath{\leftrightharpoons}}

\providecommand{\keywords}[1]{\textsc{\textsc{Keywords:}} #1}

\begin{document}
\title{Complexity and Expressivity of Branching- \\ and
  Alternating-Time Temporal Logics \\ with Finitely Many
  Variables\footnote{Prefinal version of the paper published in Bernd
    Fischer and Tarmo Uustalu (eds.)  \textit{Theoretical Aspects of
      Computing -- ICTAC 2018}, Lecture Notes in Computer Science,
    Vol. 11187, Springer 2018, pp. 396--414. DOI
    \texttt{https://doi.org/10.1007/978-3-030-02508-3\_21}}}
    
%
\author[1]{Mikhail Rybakov}
\author[2]{Dmitry Shkatov}
\affil[1]{Tver State University and University of the Witwatersrand,
  Johannesburg}
\affil[2]{University of the Witwatersrand, Johannesburg}

\maketitle

\begin{abstract}
  We show that Branching-time temporal logics {\bf CTL} and
  ${\bf CTL}^\ast$, as well as Alternating-time temporal logics {\bf
    ATL} and ${\bf ATL}^\ast$, are as semantically expressive in the
  language with a single propositional variable as they are in the
  full language, i.e., with an unlimited supply of propositional
  variables.  It follows that satisfiability for {\bf CTL}, as well as
  for \ATL, with a single variable is EXPTIME-complete, while
  satisfiability for ${\bf CTL}^\ast$, as well as for \ATLstar, with a
  single variable is 2EXPTIME-complete,---i.e., for these logics, the
  satisfiability for formulas with only one variable is as hard as
  satisfiability for arbitrary formulas.

\end{abstract}
\section{Introduction}

The propositional Branching-time temporal logics {\bf
  CTL}~\cite{CE81,DGL16} and \CTLstar\linebreak \cite{EH86,DGL16}
have for a long time been used in formal specification and
verification of (parallel) non-terminating computer
programs~\cite{HR04,DGL16}, such as (components of) operating systems,
as well as in formal specification and verification of hardware.  More
recently, Alternating-time temporal logics {\bf ATL} and
\ATLstar~\cite{AHK02,DGL16} have been used for formal specification
and verification of multi-agent~\cite{SLB08} and, more broadly,
so-called open systems, i.e., systems whose correctness depends on the
actions of external entities, such as the environment or other agents
making up a multi-agent system.

Logics \CTL, \CTLstar, \ATL, and \ATLstar\ have two main applications
to computer system design, corresponding to two different stages in
the system design process, traditionally conceived of as having
specification, implementation, and verification phases.  First, the
task of verifying that an implemented system conforms to a
specification can be carried out by checking that a formula expressing
the specification is satisfied in the structure modelling the
system,---for program verification, this structure usually models
execution paths of the program; this task corresponds to the model
checking problem~\cite{CGP00} for the logic. Second, the task of
verifying that a specification of a system is satisfiable---and, thus,
can be implemented by some system---corresponds to the satisfiability
problem for the logic.  Being able to check that a specification is
satisfiable has the obvious advantage of avoiding wasted effort in
trying to implement unsatisfiable systems.  Moreover, an algorithm
that checks for satisfiability of a formula expressing a specification
builds, explicitly or implicitly, a model for the formula, thus
supplying a formal model of a system conforming to the specification;
this model can subsequently be used in the implementation phase.
There is hope that one day such models can be used as part of a
``push-button'' procedure producing an assuredly correct
implementation from a specification model, avoiding the need for
subsequent verification altogether.  Tableaux-style
satisfiability-checking algorithms developed for \CTL\ in~\cite{EH85},
for \CTLstar\ in~\cite{Reynolds09}, for \ATL\ in~\cite{GorSh09}, and
for \ATLstar\ in~\cite{David15} all implicitly build a model for the
formula whose satisfiability is being checked.

In this paper, we are concerned with the satisfiability problem for
{\bf CTL}, \CTLstar, {\bf ATL}, and \ATLstar; clearly, the complexity
of satisfiability for these logics is of crucial importance to their
applications to formal specification. It is well-known that, for
formulas that might contain contain an arbitrary number of
propositional variables, the complexity of satisfiability for all of
these logics is quite high: it is EXPTIME-complete for {\bf
  CTL}~\cite{FL79,EH85}, 2EXPTIME-complete for
\CTLstar~\cite{VardiStockmeyer85}, EXPTIME-complete for {\bf
  ATL}~\cite{GD06,WLWW06}, and 2EXPTIME-complete for
\ATLstar~\cite{Schewe08}.

It has, however, been observed (see, for example, \cite{DS02}) that,
in practice, formulas expressing formal specifications, despite being
quite long and containing deeply nested temporal operators, usually
contain only a very small number of propositional
variables,---typically, two or three.  The question thus arises
whether limiting the number of propositional variables allowed to be
used in the construction of formulas we take as inputs can bring down
the complexity of the satisfiability problem for {\bf CTL}, \CTLstar,
{\bf ATL}, and \ATLstar.  Such an effect is not, after all, unknown in
logic:
examples are known of logics whose satisfiability problem goes down
from ``intractable'' to ``tractable'' once we place a limit on the
number of propositional variables allowed in the language: thus,
satisfiability for the classical propositional logic as well as the
extensions of the modal logic {\bf K5}~\cite{NTh75}, which include
such logics as {\bf K45}, {\bf KD45}, and {\bf S5} (see
also~\cite{Halpern95}), goes down from NP-complete to polynomial-time
decidable once we limit the number of propositional variables in the
language to an (arbitrary) finite number.\footnote{To avoid ambiguity,
  we emphasise that we use the standard complexity-theoretic
  convention of measuring the complexity of the input as its size; in
  our case, this is the length of the input formula.  In other words,
  we do not measure the complexity of the input according to how many
  distinct variables it contains; limiting the number of variables
  simply provides a restriction on the languages we consider.}
Similarly, as follows from~\cite{Nishimura60}, satisfiability for the
intuitionistic propositional logic goes down from PSPACE-complete to
polynomial-time decidable if we allow only a single propositional
variable in the language.

The question of whether the complexity of satisfiability for {\bf
  CTL}, \CTLstar, {\bf ATL}, and \ATLstar can be reduced by
restricting the number of propositional variables allowed to be used
in the formulas has not, however, been investigated in the literature.
The present paper is mostly meant to fill that gap.

A similar question has been answered in the negative for Linear-time
temporal logic {\bf LTL} in~\cite{DS02}, where it was shown, using a
proof technique peculiar to {\bf LTL} (in particular, \cite{DS02}
relies on the fact that for {\bf LTL} with a finite number of
propositional variables satisfiability reduces to model-checking),
that a single-variable fragment of {\bf LTL} is PSPACE-complete, i.e.,
as computationally hard as the entire logic~\cite{SislaClarke85}. It
should be noted that, in this respect, {\bf LTL} behaves like most
``natural'' modal and temporal logics, for which the presence of even
a single variable in the language is sufficient to generate a fragment
whose satisfiability is as hard as satisfiability for the entire
logic.  The first results to this effect have been proven
in~\cite{BS93} for logics for reasoning about linguistic structures
and in~\cite{Sve03} for provability logic.  A general method of
proving such results for PSPACE-complete logics has been proposed
in~\cite{Halpern95}; even though~\cite{Halpern95} considers only a
handful of logics, the method can be generalised to large classes of
logics, often in the language without propositional
variables~\cite{Hem01,ChRyb03} (it is not, however, applicable to {\bf
  LTL}, as it relies on unrestricted branching in the models of the
logic, which runs contrary to the semantics of {\bf LTL},---hence the
need for a different approach, as in~\cite{DS02}).  In this paper, we
use a suitable modification of the technique from~\cite{Halpern95}
(see~\cite{RSh18a,RSh18b}) to show that single-variable fragments of
{\bf CTL}, \CTLstar, {\bf ATL}, and \ATLstar\ are as computationally
hard as the entire logics; thus, for these logics, the complexity of
satisfiability cannot be reduced by restricting the number of
variables in the language.

Before doing so, a few words might be in order to explain why the
technique from~\cite{Halpern95} is not directly applicable to the
logics we are considering in this paper. The approach
of~\cite{Halpern95} is to model propositional variables by (the
so-called pp-like) formulas of a single variable; to establish the
PSPACE-harness results presented in~\cite{Halpern95}, a substitution
is made of such pp-like formulas for propositional variables into
formulas encoding a PSPACE-hard problem.  In the case of logics
containing modalities corresponding to transitive relations, such as
the modal logic {\bf S4}, for such a substitution to work, the
formulas into which the substitution is made need to satisfy the
property referred to in~\cite{Halpern95} as ``evidence in a
structure,''---a formula is evident in a structure if it has a model
satisfying the following heredity condition: if a propositional
variable is true at a state, it has to be true at all the states
accessible from that state. In the case of PSPACE-complete logics,
formulas satisfying the evidence condition can always be found, as the
intuitionistic logic, which is PSPACE-complete, has the heredity
condition built into its semantics.  The situation is drastically
different for logics that are EXPTIME-hard, which is the case for all
the logics considered in the present paper: to show that a logic is
EXPTIME-hard, one uses formulas that require for their satisfiability
chains of states of the length exponential in the size of the
formula,---this cannot be achieved with formulas that are evident in a
structure, as by varying the valuations of propositional variables
that have to satisfy the heredity condition we can only describe
chains whose length is linear in the size of the formula.  Thus, the
technique from~\cite{Halpern95} is not directly applicable to
EXPTIME-hard logics with ``transitive'' modalities, as the formulas
into which the substitution of pp-like formulas needs to be made do
not satisfy the condition that has to be met for such a substitution
to work. As all the logics considered in this paper do have a
``transitive'' modality---namely, the temporal connective ``always in
the future'', which is interpreted by the reflexive, transitive
closure of the relation corresponding to the temporal connective ``at
the next instance''---this limitation prevents the technique
from~\cite{Halpern95} from being directly applied to them.

In the present paper, we modify the approach of~\cite{Halpern95} by
coming up with substitutions of single-variable formulas for
propositional variables that can be made into arbitrary formulas,
rather than formulas satisfying a particular property, such as
evidence in a structure.  This allows us to break away from the class
PSPACE and to deal with {\bf CTL}, \CTLstar, {\bf ATL}, and \ATLstar,
all of which are at least EXPTIME-hard.  A similar approach has
recently been used in~\cite{RSh18a} and~\cite{RSh18b} for some other
propositional modal logics.

A by-product of our approach, and another contribution of this paper,
is that we establish that single-variable fragments of {\bf CTL},
\CTLstar, {\bf ATL}, and \ATLstar\ are as semantically expressive as
the entire logic, i.e., all properties that can be specified with any
formula of the logic can be specified with a formula containing only
one variable---indeed, our complexity results follow from this. In
this light, the observation cited above---that in practice most
properties of interest are expressible in these logics using only a
very small number of variables---is not at all surprising from a
purely mathematical point of view, either.


The paper is structured as follows. In
Section~\ref{sec:syntax-semantics}, we introduce the syntax and
semantics of {\bf CTL} and \CTLstar. Then, in
Section~\ref{sec:single-variable-fragment}, we show that {\bf CTL} and
\CTLstar\ can be polynomial-time embedded into their single-variable
fragments.  As a corollary, we obtain that satisfiability for the
single variable fragment of {\bf CTL} is EXPTIME-complete and
satisfiability for the single variable of of \CTLstar\ is
2EXPTIME-complete. In Section~\ref{sec:atl}, we introduce the syntax
and semantics of {\bf ATL} and \ATLstar.  Then, in
Section~\ref{sec:alt-finite-variable-fragments}, we prove results for
\ATL\ and \ATLstar\ that are analogous to those proven in
Section~\ref{sec:single-variable-fragment} for \CTL\ and \CTLstar.  We
conclude in Section~\ref{sec:conclusion} by discussing other
formalisms related to the logics considered in this paper to which our
proof technique can be applied to obtain similar results.


\section{Branching-time temporal logics}
\label{sec:syntax-semantics}

We start by briefly recalling the syntax and semantics of \CTL\ and
\CTLstar.

The language of ${\bf CTL}^\ast$ contains a countable set
$\Var = \{p_1, p_2, \ldots \}$ of propositional variables, the
propositional constant $\falsehood$ (``falsehood''), the Boolean
connective $\imp$ (``if \ldots, then \ldots''), the path quantifier
$\forall$, and temporal connectives $\next$ (``next'') and $\until$
(``until'').  The language contains two kinds of formulas: state
formulas and path formulas, so called because they are evaluated in
the models at states and paths, respectively. State formulas $\vp$ and
path formulas $\vartheta$ are simultaneously defined by the following
BNF expressions:
\[\vp ::= p  \mid \falsehood \mid  (\vp
\imp \vp) \mid \forall \vartheta, \]
\[\vartheta ::= \vp \mid (\vartheta \imp \vartheta) \mid (\vartheta \until \vartheta)
\mid \next \vartheta, \]
where $p$ ranges over \Var.  Other Boolean connectives are defined as
follows: $\neg A := (A \imp \falsehood)$,
$(A \con B) := \neg (A \imp \neg B)$, $(A \dis B) := (\neg A \imp B)$,
and $(A \equivalence B) := (A \imp B) \con (B \imp A)$, where $A$ and
$B$ can be either state or path formulas.  We also define
$\top := {\bottom} \imp {\bottom}$,
$\Diamond\, \vartheta := (\truth \until \vartheta)$,
$\Box\, \vartheta := \neg \Diamond \neg \vartheta$, and
$\exists\, \vartheta := \neg \forall \neg \vartheta$.

Formulas are evaluated in Kripke models.  A Kripke model is a tuple
\linebreak $\mmodel{M} = (\states{S}, \ar, V)$, where \states{S} is a
non-empty set (of states), $\ar$ is a binary \linebreak (transition)
relation on \states{S} that is serial (i.e., for every
$s \in \states{S}$, there exists $s' \in \states{S}$ such that
$s \ar s'$), and $V$ is a (valuation) function
$V: \Var \rightarrow 2^{\states{S}}$.

An infinite sequence $s_0, s_1, \ldots$ of states in \mmodel{M} such
that $s_i \ar s_{i+1}$, for every $i \geqslant 0$, is called a
\textit{path}.  Given a path $\pi$ and some $i \geqslant 0$, we denote
by $\pi[i]$ the $i$th element of $\pi$ and by $\pi[i, \infty]$ the
suffix of $\pi$ beginning at the $i$th element. If $s \in \states{S}$,
we denote by $\Pi(s)$ the set of all paths $\pi$ such that
$\pi[0] = s$.

The satisfaction relation between models \mmodel{M}, states $s$, and
state formulas $\vp$, as well as between models \mmodel{M}, paths
$\pi$, and path formulas $\vartheta$, is defined as follows:
\begin{itemize}
\item \sat{M}{s}{p_i} \sameas\ $s \in V(p_i)$;
\nopagebreak[3]
\item \sat{M}{s}{\falsehood} never holds;
\nopagebreak[3]
\item \sat{M}{s}{\vp_1 \imp \vp_2} \sameas\ \sat{M}{s}{\vp_1} implies
  \sat{M}{s}{\vp_2};
\nopagebreak[3]
\item \sat{M}{s}{\forall \vartheta_1} \sameas\ \sat{M}{\pi}{\vartheta_1} for
  every $\pi \in \Pi(s)$.
\nopagebreak[3]
\item \sat{M}{\pi}{\vp_1} \sameas\ \sat{M}{\pi[0]}{\vp_1};
  \nopagebreak[3]
\item \sat{M}{\pi}{\vartheta_1 \imp \vartheta_2} \sameas\ \sat{M}{\pi}{\vartheta_1} implies
  \sat{M}{\pi}{\vartheta_2};
\nopagebreak[3]
\item \sat{M}{\pi}{\next \vartheta_1} \sameas\ \sat{M}{\pi[1, \infty]}{\vartheta_1};
\nopagebreak[3]
\item \sat{M}{\pi}{\vartheta_1 \until \vartheta_2} \sameas\ \sat{M}{\pi[i,
    \infty]}{\vartheta_2} for some $i \geqslant 0$ and \sat{M}{\pi[j,
    \infty]}{\vartheta_1} for every $j$ such that $0 \leqslant j < i$.
\end{itemize}
A \CTLstar-formula is a state formula in this language. A
\CTLstar-formula is satisfiable if it is satisfied by some state of
some model, and valid if it is satisfied by every state of every
model.  Formally, by $\CTLstar$ we mean the set of valid
\CTLstar-formulas.  Notice that this set is closed under uniform
substitution.

Logic \CTL\ can be thought of as a fragment of \CTLstar\ containing
only formulas where a path quantifier is always paired up with a
temporal connective.  This, in particular, disallows formulas whose
main sign is a temporal connective and, thus, eliminates
path-formulas.  Such composite ``modal'' operators are $\AX$
(universal ``next''), $\forall\, \until$ (universal ``until''), and
$\exists\, \until$ (existential ``until'').  Formulas are defined by
the following BNF expression:
\[\vp ::= p  \mid \falsehood \mid  (\vp
\imp \vp) \mid \AX \vp \mid \forall\, (\vp \until \vp) \mid \exists\,
(\vp \until \vp),
\] where $p$ ranges over \Var. We also define
$\neg \vp := (\vp \imp \falsehood)$,
$(\vp \con \psi) := \neg (\vp \imp \neg \psi)$,
$(\vp\dis \psi) := (\neg \vp \imp \psi)$,
$\top = {\bottom} \imp {\bottom}$, $\EX \vp := \neg \AX \neg \vp$,
$\exists \Diamond \vp := \exists ( \truth \until \vp)$, and
$\forall \Box \vp := \neg \exists \Diamond \neg \vp$.

The satisfaction relation between models \mmodel{M}, states $s$, and
formulas $\vp$ is inductively defined as follows (we only list the
cases for the ``new'' modal operators):
\begin{itemize}
\item \sat{M}{s}{\AX \vp_1} \sameas\ \sat{M}{s'}{\vp_1} whenever $s \ar s'$;
\nopagebreak[3]
\item \sat{M}{s}{\forall (\vp_1 \until \vp_2)} \sameas\ for every path
  $s_0 \ar s_1 \ar \ldots$ with $s_0 = s$, \linebreak
  \sat{M}{s_i}{\vp_2}, for some $i \geqslant 0$, and
  \sat{M}{s_j}{\vp_1}, for every $0 \leqslant j < i$; \nopagebreak[3]
\item \sat{M}{s}{\exists (\vp_1 \until \vp_2)} \sameas\ there exists a
  path $s_0 \ar s_1 \ar \ldots$ with $s_0 = s$, such that
  \sat{M}{s_i}{\vp_2}, for some $i \geqslant 0$, and
  \sat{M}{s_j}{\vp_1}, for every $0 \leqslant j < i$.
\end{itemize}
Satisfiable and valid formulas are defined as for \CTLstar.  Formally,
by $\CTL$ we mean the set of valid \CTL-formulas; this set is closed
under substitution.

For each of the logics described above, by a variable-free fragment we
mean the subset of the logic containing only formulas without any
propositional variables.  Given formulas $\vp$, $\psi$ and a
propositional variable $p$, we denote by $\vp[p/\psi]$ the result of
uniformly substituting $\psi$ for $p$ in $\vp$.

\section{Finite-variable fragments of \CTLstar\ and \CTL}
\label{sec:single-variable-fragment}

In this section, we consider the complexity of satisfiability for
finite-variable fragments of {\bf CTL} and \CTLstar, as well as
semantic expressivity of those fragments.

We start by noticing that for both {\bf CTL} and \CTLstar\
satisfiability of the variable-free fragment is polynomial-time
decidable.  Indeed, it is easy to check that, for these logics, every
variable-free formula is equivalent to either $\bottom$ or $\top$.
%
Thus, to check for satisfiability of a variable-free formula $\vp$,
all we need to do is to recursively replace each subformula of $\vp$
by either $\bottom$ or $\top$, which gives us an algorithm that runs
in time linear in the size of $\vp$.  Since both {\bf CTL} and
\CTLstar\ are at least EXPTIME-hard and P $\ne$ EXPTIME, variable-free
fragments of these logics cannot be as expressive as the entire logic.

We next prove that the situation changes once we allow just one
variable to be used in the construction of formulas.  Then, we can
express everything we can express in the full languages of \CTL\ and
\CTLstar; as a consequence, the complexity of satisfiability becomes
as hard as satisfiability for the full languages. In what follows, we
first present the proof for \CTLstar, and then point out how that work
carries over to {\bf CTL}.

Let $\vp$ be an arbitrary \CTLstar-formula.  Without a loss of
generality we may assume that $\vp$ contains propositional variables
$p_1, \ldots p_n$. Let $p_{n+1}$ be a variable not occurring in $\vp$.
First, inductively define the translation $\cdot'$ as follows:
\begin{center}
\begin{tabular}{llll}
  ${p_i}'$ & = & $p_i$, $\mbox{~~where~} i \in \{1, \ldots, n \}$; \\
  $\bottom'$ & = & $\bottom$; & \\
  $(\phi \imp \psi)'$ & = & $\phi' \imp \psi'$; & \\
  $(\forall \alpha)'$ & = & $\forall (\Box p_{n+1} \imp \alpha')$; &\\
  $(\next \alpha)'$ & = & $\next \alpha'$; & \\
  $(\alpha \until \beta)'$ & = & $\alpha' \until \beta'$.  &
\end{tabular}
\end{center}
Next, let
$$ \Theta = p_{n+1} \con \AG (\EX p_{n+1} \equivalence
p_{n+1}),$$
and define $$ \widehat{\vp} = \Theta \con \varphi'.$$

Intuitively, the translation $\cdot'$ restricts evaluation of formulas
to the paths where every state makes the variable $p_{n+1}$ true,
while $\Theta$ acts as a guard making sure that all paths in a model
satisfy this property.  Notice that $\vp$ is equivalent to
$\widehat{\vp} [p_{n+1} / \top]$.

\begin{lemma}
  \label{lem:varphi-truth}
  Formula $\varphi$ is satisfiable if, and only if, formula
  $\widehat{\vp}$ is satisfiable.
\end{lemma}

\begin{proof}
  Suppose that $\widehat{\vp}$ is not satisfiable.  Then,
  $\neg \widehat{\vp} \in \CTLstar$ and, since $\CTLstar$ is closed
  under substitution,
  $\neg \widehat{\vp} [p_{n+1} / \top] \in \CTLstar$.  As
  $\widehat{\vp} [p_{n+1} / \top] \equivalence \vp \in \CTLstar$, so
  $\neg \vp \in \CTLstar$; thus, $\vp$ is not satisfiable.

  Suppose that $\widehat{\vp}$ is satisfiable. In particular, let
  $\sat{M}{s_0}{\widehat{\vp}}$ for some model $\mmodel{M}$ and some
  $s_0$ in $\mmodel{M}$. Define $\mmodel{M}'$ to be the smallest
  submodel of $\mmodel{M}$ such that
  \begin{itemize}
  \item $s_0$ is in $\mmodel{M'}$;
  \item if $x$ is in $\mmodel{M'}$, $x \ar y$, and
    \sat{M}{y}{p_{n+1}}, then $y$ is also in $\mmodel{M'}$.
  \end{itemize}
  Notice that, since \sat{M}{s_0}{p_{n+1} \con \AG (\EX p_{n+1}
    \equivalence p_{n+1})}, the model $\mmodel{M'}$ is serial, as
  required, and that $p_{n+1}$ is true at every state of
  $\mmodel{M}'$.

  We now show that $\sat{M'}{s_0}{\vp}$.  Since
  \sat{M}{s_0}{\varphi'}, it suffices to prove that, for every state
  $x$ in \mmodel{M'} and every state subformula $\psi$ of $\varphi$,
  we have \sat{M}{x}{\psi'} if, and only if, \sat{M'}{x}{\psi}; and
  that, for every path $\pi$ in \mmodel{M'} and every path subformula
  $\alpha$ of $\varphi$, we have \sat{M}{\pi}{\alpha'} if, and only
  if, \sat{M'}{\pi}{\alpha}. This can be done by simultaneous
  induction on $\psi$ and $\alpha$.

  The base case as well as Boolean cases are straightforward.

  Let $\psi = \forall \alpha$, so
  $\psi' = \forall (\Box p_{n+1} \imp \alpha')$.  Assume that
  \notsat{M}{x}{\forall (\Box p_{n+1} \imp \alpha')}. Then,
  \notsat{M}{\pi}{\alpha'}, for some $\pi \in \Pi(x)$ such that
  \sat{M}{\pi[i]}{p_{n+1}}, for every $i \geqslant 0$.  By
  construction of \mmodel{M'}, $\pi$ is a path is \mmodel{M'}; thus,
  we can apply the inductive hypothesis to conclude that
  \notsat{M'}{\pi}{\alpha}.  Therefore, \notsat{M'}{x}{\forall
    \alpha}, as required.  Conversely, assume that
  \notsat{M'}{x}{\forall \alpha}.  Then, \notsat{M'}{\pi}{\alpha}, for
  some $\pi \in \Pi(x)$.  Clearly, $\pi$ is a path in \mmodel{M}.
  Since $p_{n+1}$ is true at every state in \mmodel{M'}, and thus, at
  every state in $\pi$, using the inductive hypothesis, we conclude
  that \linebreak \notsat{M}{x}{\forall (\Box p_{n+1} \imp \alpha')}.

  The cases for the temporal connectives are straightforward.

\end{proof}

\begin{lemma}
  \label{lem:pn+1}
  If $\widehat{\vp}$ is satisfiable, then it is satisfied in a model
  where $p_{n+1}$ is true at every state.
\end{lemma}

\begin{proof}
  If $\widehat{\vp}$ is satisfiable, then, as has been shown in the
  proof of Lemma~\ref{lem:varphi-truth}, $\vp$ is satisfied in a model
  where $p_{n+1}$ is true at every state; i.e., \sat{M}{s}{\vp} for
  some $\mmodel{M} = (\states{S}, \ar, V)$ such that $p_{n+1}$ is true
  at every state in \states{S} and some $s \in S$. Since $\vp$ is
  equivalent to $\widehat{\vp} [p_{n+1} / \top]$, clearly
  \sat{M}{s}{\widehat{\vp}}.
\end{proof}

Next, we model all the variables of $\widehat{\vp}$ by single-variable
formulas $A_1, \ldots, A_m$.  This is done in the following way.
Consider the class \cm{M} of models that, for each
$m \in \{1, \ldots, n+1\}$, contains a model
$\mmodel{M}_m = (\states{S}_m, \ar, V_m)$ defined as follows:
\begin{itemize}
  \item $\states{S}_m = \{r_m, b^m, a_1^m, a_2^m, \ldots, a_{2m}^m\}$;
  \item
    $\ar\ = \{ \langle r_m, b^m \rangle, \langle r_m, a_1^m \rangle \}
    \union \{ \langle a_i^m, a_{i+1}^m \rangle : 1 \leq m \leq 2m -
    1 \} \union \\ \{ \langle s, s \rangle : s \in \states{S}_m \}$;
  \item $ s \in V_m (p)$ if, and only if, $s = r_m$ or $s = a_{2k}^m$,
    for some $k \in \{1, \ldots, m \}$.
  \end{itemize}

\begin{figure}
\begin{center}
\begin{picture}(100,195)
\put(60,10){$\circ$}
\put(60,50){$\circ$}
\put(60,90){$\circ$}
\put(60,130){$\circ$}
\put(60,190){$\circ$}
\put(20,50){$\circ$}
\put(70,10){$r_m\models A_m$}
\put(45,10){$~p$}
\put(45,10){}
\put(42,50){$\neg p$}
\put(45,90){$~p$}
\put(42,130){$\neg p$}
\put(45,190){$~p$}
\put(2,50){$\neg p$}
\put(70,10){}
\put(70,50){$a_1^m$}
\put(70,90){$a_2^m$}
\put(70,130){$a_3^m$}
\put(70,190){$a_{2m}^m$}
\put(20,60){$b^m$}
\put(62.1,14){\vector(0,1){36}}
\put(62.1,54){\vector(0,1){36}}
\put(62.1,94){\vector(0,1){36}}
\put(61,14){\vector(-1,1){37}}
\put(62.1,134){\vector(0,1){19}}
\put(62.1,172){\vector(0,1){19}}
\put(61.1,158){$\vdots$}
\end{picture}
\caption{Model $\frak{M}_m$}
\label{fig_Mm_ctl_star}
\end{center}
\end{figure}
  
  The model $\mmodel{M}_m$ is depicted in
  Figure~\ref{fig_Mm_ctl_star}, where circles represent states with
  loops.
With every such $\mmodel{M}_m$, we associate a formula $A_m$, in the
following way.  First, inductively define the sequence of formulas
  \begin{center}
  \begin{tabular}{lll}
    $\chi_0$ & = & $\forall\, \Box\, p$; \\
    $\chi_{k+1}$ & = & $p \con \EX ( \neg p \con \EX \chi_k)$.
  \end{tabular}
  \end{center}
Next, for every $m \in \{1, \ldots, n + 1\}$, let
$$
A_m = \chi_m \con \EX \AG \neg p.
$$

\begin{lemma}
  \label{lem:roots}
  Let $\mmodel{M}_k \in \cm{M}$ and let $x$ be a state in
  $\mmodel{M}_k$. Then, $\mmodel{M}_k, x \models A_m$ if, and only if,
  $k = m$ and $x = r_m$.
\end{lemma}

\begin{proof}
  Straightforward.
\end{proof}

\noindent Now, for every $m \in \{1, \ldots, n+1\}$, define 
$$B_m = \EX A_m.$$ Finally, let $\sigma$ be a (substitution) function
that, for every $i \in \{1 \ldots n + 1\}$, replaces $p_i$ by $B_i$,
and let
$$\varphi^* = \sigma (\widehat{\vp}).$$
Notice that the formula $\varphi^*$ contains only a single variable, $p$.

\begin{lemma}
  \label{lem:main_lemma_ctl_star}
  Formula $\varphi$ is satisfiable if, and only if, formula
  $\varphi^*$ is satisfiable.
\end{lemma}

\begin{proof}
  Suppose that $\varphi$ is not satisfiable.  Then, in view of
  Lemma~\ref{lem:varphi-truth}, $\widehat{\vp}$ is not satisfiable.
  Then, $\neg \widehat{\vp} \in \CTLstar$ and, since
  $\CTLstar$ is closed under substitution,
  $\neg \vp^\ast \in \CTLstar$.  Thus, $\vp^\ast$ is not
  satisfiable.
  
  Suppose that $\varphi$ is satisfiable.  Then, in view of
  Lemmas~\ref{lem:varphi-truth} and~\ref{lem:pn+1}, $\widehat{\vp}$ is
  satisfiable in a model $\mmodel{M} = (\states{S}, \ar, V)$ where
  $p_{n+1}$ is true at every state.  We can assume without a loss of
  generality that every $x \in \states{S}$ is connected by some path
  to $s$. Define model \mmodel{M'} as follows. Append to \mmodel{M}
  all the models from \cm{M} (i.e., take their disjoint union), and
  for every $x \in \states{S}$, make $r_m$, the root of
  $\mmodel{M}_m$, accessible from $x$ in \mmodel{M'} exactly when
  \sat{M}{x}{p_m}.  The evaluation of $p$ is defined as follows: for
  states from each $\mmodel{M}_m \in \cm{M}$, the evaluation is the
  same as in $\mmodel{M}_m$, and for every $x \in \states{S}$, let
  $x \notin V'(p)$.

  We now show that \sat{M'}{s}{\varphi^*}.  It is easy to check that
  \sat{M'}{s}{\sigma(\Theta)}.  It thus remains to show that
  \sat{M'}{s}{\sigma(\varphi')}.  Since \sat{M}{s}{\varphi'}, it
  suffices to prove that \sat{M}{x}{\psi'} if, and only if,
  \sat{M'}{x}{\sigma(\psi')}, for every state $x$ in \mmodel{M} and
  every state subformula $\psi$ of $\varphi$; and that
  \sat{M}{\pi}{\alpha'} if, and only if,
  \sat{M'}{\pi}{\sigma(\alpha')}, for every path $\pi$ in \mmodel{M}
  and every path subformula $\alpha$ of $\varphi$. This can be done by
  simultaneous induction on $\psi$ and $\alpha$.

  Let $\psi = p_i$, so $\psi' = p_i$ and $\sigma(\psi') = B_i$.
  Assume that \sat{M}{x}{p_i}.  Then, by construction of \mmodel{M'},
  we have \sat{M'}{x}{B_i}.  Conversely, assume that \sat{M'}{x}{B_i}.
  As \sat{M'}{x}{B_i} implies \sat{M'}{x}{\EX p} and since
  \notsat{M}{y}{p}, for every $y \in \states{S}$, this can only happen
  if $x \ar^{\mmodel{M'}} r_m$, for some $m \in \{1, \ldots, n+1\}$.
  Since, then, $r_m \models A_i$, in view of Lemma~\ref{lem:roots},
  $m = i$, and thus, by construction of \mmodel{M'}, we have
  \sat{M}{x}{p_i}.

  The Boolean cases are straightforward.

  Let $\psi = \forall \alpha$, so
  $\psi' = \forall (\Box p_{n+1} \imp \alpha')$ and
  $\sigma(\psi') = \forall (\Box B_{n+1} \imp \sigma(\alpha'))$.
  Assume that \notsat{M}{x}{\forall (\Box p_{n+1} \imp
    \alpha')}. Then, for some $\pi \in \Pi(x)$ such that
  \sat{M}{\pi[i]}{p_{n+1}} for every $i \geqslant 0$, we have
  \notsat{M}{\pi}{\alpha'}.  Clearly, $\pi$ is a path in \mmodel{M'},
  and thus, by inductive hypothesis, \sat{M'}{\pi[i]}{B_{n+1}}, for
  every $i \geqslant 0$, and \notsat{M'}{\pi}{\sigma(\alpha')}.  Hence,
  \notsat{M'}{x}{\forall (\Box B_{n+1} \imp \sigma(\alpha'))}, as
  required.  Conversely, assume that \notsat{M'}{x}{\forall (\Box
    B_{n+1} \imp \sigma(\alpha'))}.  Then, for some $\pi \in \Pi(x)$
  such that \sat{M'}{\pi[i]}{B_{n+1}} for every $i \geqslant 0$, we
  have \notsat{M'}{\pi}{\sigma(\alpha')}.  Since by construction of
  $\mmodel{M}'$, no state outside of $\states{S}$ satisfies $B_{n+1}$,
  we know that $\pi$ is a path in \mmodel{M}.  Thus, we can use the
  inductive hypothesis to conclude that \notsat{M}{x}{\forall (\Box
    p_{n+1} \imp \alpha')}.

  The cases for the temporal connectives are straightforward.

\end{proof}

Lemma~\ref{lem:main_lemma_ctl_star}, together with the observation
that the formula $\vp^\ast$ is polynomial-time computable from $\vp$,
give us the following:

\begin{theorem}
  \label{thr:CTLstar}
  There exists a polynomial-time computable function $e$ assigning to
  every \CTLstar-formula $\vp$ a single-variable formula $e(\vp)$ such
  that $e(\vp)$ is satisfiable if, and only if, $\vp$ is satisfiable.
\end{theorem}

\begin{theorem}
  \label{thr:ctlstar-complexity}
  The satisfiability problem for the single-variable fragment of
  \CTLstar\ is\/ {\rm 2EXPTIME}-complete.
\end{theorem}

\begin{proof}
  The lower bound immediately follows from Theorem~\ref{thr:CTLstar}
  and 2EXPTIME-hardness of satisfiability for
  \CTLstar~\cite{VardiStockmeyer85}.  The upper bound follows from the
  \linebreak 2EXPTIME upper bound for satisfiability for
  \CTLstar~\cite{VardiStockmeyer85}.
\end{proof}

We now show how the argument presented above for \CTLstar\ can be
adapted to \CTL.  First, we notice that if our sole purpose were to
prove that satisfiability for the single-variable fragment of \CTL\ is
EXPTIME-complete, we would not need to work with the entire set of
connectives present in the language of \CTL,---it would suffice to
work with a relatively simple fragment of \CTL\ containing the modal
operators $\AX$ and $\AG$, whose satisfiability, as follows
from~\cite{FL79}, is EXPTIME-hard.  We do, however, also want to
establish that the single-variable fragment of \CTL\ is as expressive
the entire logic; therefore, we embed the entire \CTL\ into its
single-variable fragment.  To that end, we can carry out an argument
similar to the one presented above for \CTLstar.

First, we define the translation $\cdot'$ as follows:
\begin{center}
\begin{tabular}{llll}
  ${p_i}'$ & = & $p_i$ $\mbox{~~where~} i \in \{1, \ldots, n \}$; \\
  $(\bottom)'$ & = & $\bottom$; & \\
  $(\phi \imp \psi)'$ & = & $\phi' \imp \psi'$; & \\
  $(\AX \phi)'$ & = & $\AX (p_{n+1} \imp \phi')$; & \\
  $(\forall\, (\phi \until \psi))'$ & = & $\forall\, (\phi' \until
                                          (p_{n+1} \con \psi'))$; & \\
  $(\exists\, (\phi \until \psi))'$ & = & $\exists\, (\phi' \until
                                          (p_{n+1} \con \psi'))$. &
\end{tabular}
\end{center}
Next, let
$$ \Theta = p_{n+1} \con \AG (\EX p_{n+1} \equivalence p_{n+1}). $$
and define $$\widehat{\varphi} = \Theta \con \varphi'.$$ 

Intuitively, the translation $\cdot'$ restricts the evaluation of
formulas to the states where $p_{n+1}$ is true.  Formula $\Theta$ acts
as a guard making sure that all states in a model satisfy this
property.  We can then prove the analogues of
Lemmas~\ref{lem:varphi-truth} and~\ref{lem:pn+1}.

\begin{lemma}
  \label{lem:varphi-truth-ctl}
  Formula $\varphi$ is satisfiable if, and only if, formula
  $\widehat{\vp}$ is satisfiable.
\end{lemma}

\begin{proof}
  Analogous to the proof of Lemma~\ref{lem:varphi-truth}. In the
  right-to-left direction, inductive steps for modal connectives rely
  on the fact that in a submodel we constructed every state makes the
  variable $p_{n+1}$ true.
\end{proof}

\begin{lemma}
  \label{lem:pn+1-ctl}
  If $\widehat{\vp}$ is satisfiable, then it is satisfied in a model
  where $p_{n+1}$ is true at every state.
\end{lemma}

\begin{proof}
  Analogous to the proof of Lemma~\ref{lem:pn+1}.
\end{proof}

Next, we model propositional variables $p_1, \ldots, p_{n+1}$ in the
formula $\widehat{\varphi}$ exactly as in the argument for \CTLstar,
i.e., we use formulas $A_m$ and their associated models
$\mmodel{M}_m$, where $m \in \{1, \ldots, n+1\}$.  This can be done
since formulas $A_m$ are, in fact, \CTL-formulas.
Lemma~\ref{lem:roots} can, thus, be reused for \CTL, as well.

We then define a single-variable \CTL-formula $\vp^\ast$ analogously
to the way it had been done for \CTLstar:
$$\varphi^* = \sigma (\widehat{\vp}),$$
where $\sigma$ is a (substitution) function that, for every
$i \in \{1 \ldots n + 1\}$, replaces $p_i$ by $B_i = \EX A_i$.  We can
then prove the analogue of Lemma~\ref{lem:main_lemma_ctl_star}.

\begin{lemma}
  \label{lem:main_lemma_ctl}
  Formula $\varphi$ is satisfiable if, and only if, formula
  $\varphi^*$ is satisfiable.
\end{lemma}

\begin{proof}
  Analogous to the proof of Lemma~\ref{lem:main_lemma_ctl_star}. In
  the left-to-right direction, the inductive steps for the modal
  connectives rely on the fact that the formula $B_{n+1}$ is true
  precisely at the states of the model that satisfies $\vp$.
\end{proof}

\noindent We, thus, obtain the following:

\begin{theorem}
  \label{thr:ctl}
  There exists a polynomial-time computable function $e$ assigning to
  every \CTL-formula $\vp$ a single-variable formula $e(\vp)$ such
  that $e(\vp)$ is satisfiable if, and only if, $\vp$ is satisfiable.
\end{theorem}

\begin{theorem}
  \label{thr:ctl-complexity}
  The satisfiability problem for the single-variable fragment of \CTL\
  is\/ {\rm EXPTIME}-complete.
\end{theorem}

\begin{proof}
  The lower bound immediately follows from Theorem~\ref{thr:ctl} and
  EXPTIME-hardness of satisfiability for \CTL~\cite{FL79}.  The upper
  bound follows from the EXPTIME upper bound for satisfiability for
  \CTL~\cite{EH85}.
\end{proof}

\section{Alternating-time temporal logics}
\label{sec:atl}

Alternating-time temporal logics \ATLstar\ and \ATL\ can be conceived
of as generalisations of \CTLstar\ and \CTL, respectively.  Their
models incorporate transitions occasioned by simultaneous actions of
the agents in the system rather than abstract transitions, as in
\CTLstar\ and \CTL, and we now reason about paths that can be forced
by cooperative actions of coalitions of agents, rather than just about
all ($\forall$) and some ($\exists$) paths.  We do not lose the
ability to reason about all and some paths in \ATLstar\ and \ATL,
however, so these logics are generalisations of \CTLstar\ and \CTL,
respectively.

The language of ${\bf ATL}^\ast$ contains a non-empty, finite set
$\mathbb{A}\mathbb{G}$ of names of agents (subsets of
$\mathbb{A}\mathbb{G}$ are called coalitions); a countable set
$\Var = \{p_1, p_2, \ldots \}$ of propositional variables; the
propositional constant $\falsehood$; the Boolean connective $\imp$;
coalition quantifiers $\coal{C}$, for every
$C \subseteq \mathbb{A}\mathbb{G}$; and temporal connectives $\next$
(``next''), $\until$ (``until''), and $\Box$ (``always in the
future'').  The language contains two kinds of formulas: state
formulas and path formulas. State formulas $\vp$ and path formulas
$\alpha$ are simultaneously defined by the following BNF expressions:
\[\vp ::= p  \mid \falsehood \mid  (\vp
\imp \vp) \mid \coal{C} \vartheta, \]
\[\vartheta ::= \vp \mid (\vartheta \imp \vartheta) \mid (\vartheta \until \vartheta)
\mid \next \vartheta \mid \Box \vartheta, \]
where $C$ ranges over subsets of $\mathbb{A}\mathbb{G}$ and $p$
ranges over \Var.  Other Boolean and temporal connectives are defined
as for \CTLstar.

Formulas are evaluated in concurrent game models. A concurrent game
model is a tuple
$\mmodel{M} = (\mathbb{A}\mathbb{G}, \states{S}, Act, act, \delta,
V)$, where
\begin{itemize}
\item $\mathbb{A}\mathbb{G} = \{1, \ldots, k\}$ is a finite, non-empty
  set of agents;
\item $\states{S}$ is a non-empty set of states;
\item $Act$ is a non-empty set of actions;
\item $act: \mathbb{A}\mathbb{G} \times \states{S} \mapsto 2^{Act}$ is
  an action manager function assigning a non-empty set of
  ``available'' actions to an agent at a state;
\item $\delta$ is a transition function assigning to every state
  $s \in \states{S}$ and every action profile
  $\alpha = (\alpha_1, \ldots, \alpha_k)$, where
  $\alpha_a \in act( a, s )$, for every $a \in \mathbb{A}\mathbb{G}$,
  an outcome state $\delta(s, \alpha)$;
\item $V$ is a (valuation) function
  $V: \Var \rightarrow 2^{\states{S}}$.
\end{itemize}

A few auxiliary notions need to be introduced for the definition of
the satisfaction relation.

A \textit{path} is an infinite sequence $s_0, s_1, \ldots$ of states
in \mmodel{M} such that, for every $i \geqslant 0$, the following
holds: $s_{i+1} \in \delta(s_i, \alpha)$, for some action profile
$\alpha$.  The set of all such sequences is denoted by
$\states{S}^\omega$.  The notation $\pi[i]$ and $\pi[i, \infty]$ is
used as for \CTLstar.  Initial segments $\pi[0, i]$ of paths are
called \textit{histories}; a typical history is denoted by $h$, and
its last state, $\pi[i]$, is denoted by $last(h)$.  Note that
histories are non-empty sequences of states in $\states{S}$; we denote
the set of all such sequences by $\states{S}^+$.

%

Given $s \in \states{S}$ and $C \subseteq \mathbb{A}\mathbb{G}$, a
$C$-action at $s$ is a tuple $\alpha_C$ such that \linebreak
$\alpha_C(a) \in act(a, s)$, for every $a \in C$, and $\alpha_C(a')$,
for every $a' \notin C$, is an unspecified action of agent $a'$ at $s$
(technically, a $C$-action might be thought of as an equivalence class
on action profiles determined by a vector of chosen actions for every
$a \in C$); we denote by $act( C, s)$ the set of $C$-actions at $s$.
An action profile $\alpha$ extends a $C$-action $\alpha_C$,
symbolically $\alpha_C \sqsubseteq \alpha$, if
$\alpha(a) = \alpha_C(a)$, for every $a \in C$.  The \textit{outcome
  set} of the $C$-action $\alpha_C$ at $s$ is the set of states
$out(s, \alpha_C) = \{ \delta(s, \alpha ) \mid \alpha \in
act(\mathbb{A}\mathbb{G}, s) \mbox{ and } \alpha_C \sqsubseteq
\alpha\}$.

A \textit{strategy} for an agent $a$ is a function
$str_{a} (h) : \states{S}^+ \mapsto act( a, last(h))$ assigning to
every history an action available to $a$ at the last state of the
history.  A $C$-strategy is a tuple of strategies for every $a \in C$.
The function $out(s, \alpha_C)$ can be naturally extended to the
functions $out(s, str_C)$ and $out(h, str_C)$ assigning to a given
state $s$, or more generally a given history $h$, and a given
$C$-strategy the set of states that can result from applying $str_C$
at $s$ or $h$, respectively.  The set of all paths that can result
when the agents in $C$ follow the strategy $str_C$ from a given state
$s$ is denoted by $\Pi(s, str_C)$ and defined as
$\{ \pi \in \states{S}^\omega \mid \pi[0] = s \mbox{ and} \linebreak \pi[j+1]
\in out( \pi[0, j], str_C ), \mbox{ for every } j \geqslant 0 \}$.

The satisfaction relation between models \mmodel{M}, states $s$, and
state formulas $\vp$, as well as between models \mmodel{M}, paths
$\pi$, and path formulas $\vartheta$, is defined as follows:
\begin{itemize}
\item \sat{M}{s}{p_i} \sameas\ $s \in V(p_i)$;
\nopagebreak[3]
\item \sat{M}{s}{\falsehood} never holds;
\nopagebreak[3]
\item \sat{M}{s}{\vp_1 \imp \vp_2} \sameas\ \sat{M}{s}{\vp_1} implies
  \sat{M}{s}{\vp_2}; \nopagebreak[3]
\item \sat{M}{s}{\coal{C} \vartheta_1} \sameas\ there exists a
  $C$-strategy $str_{C}$ such that \sat{M}{\pi}{\vartheta_1} holds for
  every $\pi \in \Pi(s, str_{C})$;  \nopagebreak[3]
\item \sat{M}{\pi}{\vp_1} \sameas\ \sat{M}{\pi[0]}{\vp_1};
  \nopagebreak[3]
\item \sat{M}{\pi}{\vartheta_1 \imp \vartheta_2} \sameas\
  \sat{M}{\pi}{\vartheta_1} implies \sat{M}{\pi}{\vartheta_2};
  \nopagebreak[3]
\item \sat{M}{\pi}{\next \vartheta_1} \sameas\ \sat{M}{\pi[1,
    \infty]}{\vartheta_1}; \nopagebreak[3]
\item \sat{M}{\pi}{\Box \vartheta_1} \sameas\ \sat{M}{\pi[i,
    \infty]}{\vartheta_1}, for every $i \geqslant 0$; \nopagebreak[3]
\item \sat{M}{\pi}{\vartheta_1 \until \vartheta_2} \sameas\
  \sat{M}{\pi[i, \infty]}{\vartheta_2} for some $i \geqslant 0$ and
  \sat{M}{\pi[j, \infty]}{\vartheta_1} for every $j$ such that
  $0 \leqslant j < i$.
\end{itemize}
An \ATLstar-formula is a state formula in this language. An
\ATLstar-formula is satisfiable if it is satisfied by some state of
some model, and valid if it is satisfied by every state of every
model.  Formally, by $\ATLstar$ we mean the set of all valid
\ATLstar-formulas; notice that this set is closed under uniform
substitution.

Logic \ATL\ can be thought of as a fragment of \ATLstar\ containing
only formulas where a coalition quantifier is always paired up with a
temporal connective.  This, as in the case of \CTL, eliminates
path-formulas.  Such composite ``modal'' operators are
$\coal{C} \next$, $\coal{C} \Box$, and $\coal{C} \until$.  Formulas
are defined by the following BNF expression:
\[\vp ::= p  \mid \falsehood \mid  (\vp
\imp \vp) \mid \coal{C} \next \vp \mid \coal{C} \Box \vp \mid \coal{C}
(\vp \until \vp),
\] where $C$ ranges over subsets of $\mathbb{A}\mathbb{G}$ and $p$
ranges over \Var.  The other Boolean connectives and the constant
$\top$ are defined as for \CTL.


The satisfaction relation between concurrent game models \mmodel{M},
states $s$, and formulas $\vp$ is inductively defined as follows (we
only list the cases for the ``new'' modal operators):
\begin{itemize}
\item \sat{M}{s}{\coal{C} \next \vp_1} \sameas\ there exists a
  $C$-action $\alpha_C$ such that \sat{M}{s'}{\vp_1} whenever
  $s' \in out(s, act_C)$;
  \nopagebreak[3]
\item \sat{M}{s}{\coal{C} \Box \vp_1} \sameas\ there exists a
  $C$-strategy $str_C$ such that \sat{M}{\pi[i]}{\vp_1} holds for all
  $\pi \in out(s, str_C)$ and all $i \geqslant 0$;
\nopagebreak[3]  
\item \sat{M}{s}{\coal{C} (\vp_1 \until \vp_2)} \sameas\ there exists a
  $C$-strategy $str_C$ such that, for all $\pi \in out(s, str_C)$,
  there exists $i \geqslant 0$ with \sat{M}{\pi[i]}{\vp} and
  \sat{M}{\pi[j]}{\vp} holds for every $j$ such that $0 \leqslant j < i$.

\end{itemize}
Satisfiable and valid formulas are defined as for \ATLstar.  Formally,
by $\ATL$ we mean the set of all valid \ATLstar-formulas; this set is
closed under substitution.

\begin{remark}
  We have given definitions of satisfiability and validity for
  \ATLstar\ and \ATL\ that assume that the set of all agents
  $\mathbb{A}\mathbb{G}$ present in the language is ``fixed in
  advance''.  At least two other notions of satisfiability (and, thus,
  validity) for these logics have been discussed in the literature
  (see, e.g., \cite{WLWW06})---i.e., satisfiability of a formula in a
  model where the set of all agents coincides with the set of agents
  named in the formula and satisfiability of a formula in a model
  where the set of agents is any set including the agents named in the
  formula (in this case, it suffices to consider all the agents named
  in the formula plus one extra agent).  In what follows, we
  explicitly consider only the notion of satisfiability for a fixed
  set of agents; other notions of satisfiability can be handled in a
  similar way.
\end{remark}

\section{Finite-variable fragments of \ATLstar\ and \ATL}
\label{sec:alt-finite-variable-fragments}

We start by noticing that satisfiability for variable-free fragments
of both \ATLstar\ and \ATL\ is polynomial-time decidable, using the
algorithm similar to the one outlined for \CTLstar\ and \CTL. It
follows that variable-free fragments of \ATLstar\ and \ATL\ cannot be
as expressive as entire logics.

We also notice that, as is well-known, satisfiability for \CTLstar\ is
polynomial-time reducible to satisfiability for \ATLstar\ and
satisfiability for \CTL\ is polynomial-time reducible to
satisfiability for \ATL, using the translation that replaces all
occurrences of $\forall$ by $\coal{\varnothing}$ and all occurrences
of $\exists$ by $\coal{\mathbb{A}\mathbb{G}}$.  Thus,
Theorems~\ref{thr:ctlstar-complexity} and~\ref{thr:ctl-complexity},
together with the known upper bounds~\cite{GD06,WvdH03,Schewe08},
immediately give us the following:

\begin{theorem}
  The satisfiability problem for the single-variable fragment of
  \ATLstar\ is\/ {\rm 2EXPTIME}-complete.
\end{theorem}

\begin{theorem}
  The satisfiability problem for the single-variable fragment of \ATL\
  is\/ {\rm EXPTIME}-complete.
\end{theorem}

In the rest of this section, we show that single-variable fragments of
\ATLstar\ and \ATL\ are as expressive as the entire logics by
embedding both \ATLstar\ and \ATL\ into their single-variable
fragments.  The arguments closely resemble the ones for \CTLstar\ and
\CTL, so we only provide enough detail for the reader to be able to
easily fill in the rest.

First, consider \ATLstar.  The translation $\cdot'$
is defined as for \CTLstar, except that the clause for $\forall$ is
replaced by the following:
\begin{center}
\begin{tabular}{llll}
  $(\coal{C} \alpha)'$ & = & $\coal{C} (\Box p_{n+1} \con \alpha')$. &\\
\end{tabular}
\end{center}
Next, we define
$$ \Theta = p_{n+1} \con \coal{\varnothing} \Box
(\coal{\mathbb{A}\mathbb{G}} \next p_{n+1} \equivalence p_{n+1})$$
and $$ \widehat{\vp} = \Theta \con \varphi'.$$ Then, we can prove the
analogues of Lemmas~\ref{lem:varphi-truth} and~\ref{lem:pn+1}.

We next model all the variables of $\widehat{\vp}$ by single-variable
formulas $A'_1, \ldots, A'_m$.  To that end, we use the class of
concurrent game models
$\cm{M} = \{\mmodel{M}'_1, \ldots, \mmodel{M}'_m\}$ that closely
resemble models $\mmodel{M}_1, \ldots, \mmodel{M}_m$ used in the
argument for \CTLstar. For every $\mmodel{M}'_i$, with
$i \in \{1, \ldots, m\}$, the set of states and the valuation $V$ are
the same as for $\mmodel{M}_i$; in addition, whenever $s \ar s'$ holds
in $\mmodel{M}_i$, we set $\delta(s, \alpha ) = s'$, for every action
profile $\alpha$.  The actions available to an agent $a$ at each state
of $\mmodel{M}_i$ are all the actions available to $a$ at any of the
states of the model \mmodel{M} to which we are going to attach models
$\mmodel{M}'_i$ when proving the analogue of
Lemma~\ref{lem:main_lemma_ctl_star}, as well as an extra action $d_a$
that we need to set up transitions from the states of \mmodel{M} to
the roots of $\mmodel{M}'_i$s.

With every $\mmodel{M}'_i$ we associate the formula $A'_i$.  First,
inductively define the sequence of formulas
\begin{center}
  \begin{tabular}{lll}
    $\chi'_0$ & = & $\coal{\varnothing} \Box\, p$; \\
    $\chi'_{k+1}$ & = & $p \con \coal{\mathbb{A}\mathbb{G}} \next (
                        \neg p \con \coal{\mathbb{A}\mathbb{G}} \next \chi_k)$.
  \end{tabular}
  \end{center}
Next, for every $m \in \{1, \ldots, n + 1\}$, let
$$
A'_m = \chi'_m \con \coal{\mathbb{A}\mathbb{G}} \next \coal{\varnothing} \Box\, \neg p.
$$

\begin{lemma}
  \label{lem:roots-atl}
  Let $\mmodel{M}'_k \in \cm{M}$ and let $x$ be a state in
  $\mmodel{M}'_k$. Then, $\mmodel{M}'_k, x \models A'_m$ if, and only
  if, $k = m$ and $x = r_m$.
\end{lemma}

\begin{proof}
  Straightforward.
\end{proof}

\noindent Now, for every $m \in \{1, \ldots, n+1\}$, define
$$B'_m = \coal{\mathbb{A}\mathbb{G}} \next A'_m.$$ Finally, let $\sigma$ be a (substitution) function
that, for every $i \in \{1, \ldots, n + 1\}$, replaces $p_i$ by
$B'_i$, and let
$$\varphi^* = \sigma (\widehat{\vp}).$$
This allows us to prove the analogue of Lemma~\ref{lem:main_lemma_ctl_star}.
\begin{lemma}
  \label{lem:main_lemma_atl_star}
  Formula $\varphi$ is satisfiable if, and only if, formula
  $\varphi^*$ is satisfiable.
\end{lemma}

\begin{proof}
  Analogous to the proof of Lemma~\ref{lem:main_lemma_ctl_star}.  When
  constructing the model $\mmodel{M}'$, whenever we need to connect a
  state $s$ in $\mmodel{M}$ to the root $r_i$ of $\mmodel{M}'_i$, we
  make an extra action, $d_a$, available to every agent $a$, and
  define
  $\delta(s, \langle d_a \rangle_{a \in \mathbb{A}\mathbb{G}}) = r_i$.
\end{proof}

Thus, we have the following:
\begin{theorem}
  \label{thr:ATLstar}
  There exists a polynomial-time computable function $e$ assigning to
  every \ATLstar-formula $\vp$ a single-variable formula $e(\vp)$ such
  that $e(\vp)$ is satisfiable if, and only if, $\vp$ is satisfiable.
\end{theorem}
We then can adapt the argument for \ATL\ form the one just presented
in the same way we adapted the argument for \CTL\ from the one for
\CTLstar, obtaining the following:
\begin{theorem}
  \label{thr:ATL}
  There exists a polynomial-time computable function $e$ assigning to
  every \ATL-formula $\vp$ a single-variable formula $e(\vp)$ such
  that $e(\vp)$ is satisfiable if, and only if, $\vp$ is satisfiable.
\end{theorem}

\section{Discussion}
\label{sec:conclusion}

We have shown that logics \CTLstar, \CTL, \ATLstar, and \ATL\ can be
polynomial-time embedded into their single-variable fragments; i.e.,
their single-variable fragments are as expressive as the entire
logics.  Consequently, for these logics, satisfiability is as
computationally hard when one considers only formulas of one variable
as when one considers arbitrary formulas.  Thus, the complexity of
satisfiability for these logics cannot be reduced by restricting the
number of variables allowed in the construction of formulas.

The technique presented in this paper can be applied to many other
modal and temporal logics of computation considered in the literature.
We will not here attempt a comprehensive list, but rather mention a
few examples.

The proofs presented in this paper can be extended in a rather
straightforward way to Branching- and Alternating-time
temporal-epistemic logics~\cite{HV89,WvdH03,Walther05,GorSh09a}, i.e.,
logics that enrich the logics considered in this paper with the
epistemic operators of individual, distributed, and common knowledge
for the agents.  Our approach can be used to show that single-variable
fragments of those logics are as expressive as the entire logics and
that, consequently, the complexity of satisfiability for them is as
hard (EXPTIME-hard or 2EXPTIME-hard) as for the entire logics.
Clearly, the same approach can be applied to epistemic
logics~\cite{FHMV95,GSh08,GSh09}, i.e., logics containing epistemic,
but not temporal, operators---such logics are widely used for
reasoning about distributed computation.  Our argument also applies to
logics with the so-called universal modality~\cite{GorPassy89} to
obtain EXPTIME-completeness of their variable-free fragments. The
technique presented here has also been recently used~\cite{RSh18a} to
show that propositional dynamic logics are as expressive in the
language without propositional variables as in the language with an
infinite supply of propositional variables. Since our method is
modular in the way it tackles modalities present in the language, it
naturally lends itself to modal languages combining various
modalities---a trend that has been gaining prominence for some time
now.

The technique presented in this paper can also be lifted to
first-order languages to prove undecidability results about fragments
of first-order modal and related logics,---see~\cite{RSh18c}.

We conclude by noticing that, while we have been able to overcome the
limitations of the technique from~\cite{Halpern95} described in the
introduction, our modification thereof has limitations of its own.  It
is not applicable to logics whose semantics forbids branching, such as
{\bf LTL} or temporal-epistemic logics of linear
time~\cite{HV89,GSh09b}.  Our technique cannot be used, either, to
show that finite-variable fragments of logical systems that are not
closed under uniform substitution---such as public announcement logic
{\bf PAL}~\cite{Plaza89,DHK08}---have the same expressive power as the
entire system.  This does not preclude it from being used in
establishing complexity results for finite-variable fragments of such
systems provided they contain fragments, as is the case with {\bf
  PAL}~\cite{Lutz06}, that are closed under substitution and have the
same complexity as the entire system.


\bibliographystyle{plain}

\end{document}